\documentclass[11pt]{article}

\usepackage{latexsym}
\usepackage{amsmath}
\usepackage{amssymb}

\usepackage{graphicx}

\newtheorem{theorem}{Theorem}
\newtheorem{proposition}{Proposition}

\newtheorem{definition}{Definition}

\title{Whittle index approach to multi-server scheduling \\ 
with convex delay costs and impatient customers}

\author{Samuli Aalto \\ 
Department of Information and Communications Engineering \\
Aalto University, Finland
}

\begin{document}

\date{}

\maketitle

\begin{abstract}

We consider the dynamic scheduling problem in a multi-class M/G/$N$ + M 
queue with convex delay costs and impatient customers that have exponential 
abandonment times. We apply the Whittle index approach to find a reasonable 
heuristic solution for this tricky problem. By assuming exponential 
abandonment times, we are able to make a very straightforward use of the 
results of the corresponding problem with patient customers presented in 
\cite{Aal26QS}. Our main theoretical achievements are proving that the 
closed version of the corresponding discrete-time problem is indexable 
and deriving an explicit expression for the Whittle index. The discrete-time 
results are utilized to develop the Whittle index policy for the original 
continuous-time scheduling problem.
\end{abstract}

\newpage

\section{Introduction}
\label{sec:intro}

We consider the dynamic scheduling problem in a multi-class M/G/$N$ + M queue 
with convex delay costs and impatient customers. We are interested in 
optimizing the scheduling policy among all non-anticipating\footnote{
A non-anticipating scheduling policy is aware of the attained services and 
the times that the customers have already spent in the system, but it does 
not have any knowledge of the remaining service times or the remaining 
abandonment times of the customers.} 
policies that allow preemption with an objective function that takes 
into account both the holding costs of the customers in the system 
and the abandonment penalties related to the customers that leave the 
system before their service has been completed. 
More precisely said, the holding costs for a class-$k$ customer that has 
spent $x$ time units in the system accrue with rate $\gamma_k(x)$, where 
$\gamma_k(x)$ is an increasing\footnote{
In this paper, we use the terms `increasing' and `decreasing' in their weak 
forms. Thus, e.g., an increasing function is not required to be strictly 
increasing.} 
function of $x$. Consequently, the accumulated holding costs, 
$\int_0^x \gamma_k(t) \, \mathrm{d}t$, compose a convexly increasing function 
of delay $x$.\footnote{
Note that the {\em time-convex} holding costs introduced in \cite{Mie95AAP} 
and studied also in the present paper differ from the {\em count-convex} 
holding costs, for which the aggregate holding cost rate related to 
{class-$k$} customers increases convexly as a function of the current number 
of {class-$k$} customers in the system. The optimal scheduling problem with 
the latter type of holding costs and impatient customers have been studied, 
e.g., in \cite{Lar15QS,Lon20OpnsRes}.}
In addition to the holding costs, a $d_k$ penalty fee (lump sum) 
is paid if the customer leaves due to abandonment.

If there is just a single server, the holding costs are linearly 
increasing (i.e., $\gamma_k(t) = c_k$ for all $t$), and the customers are 
patient, the well-known $c\mu$ rule \cite{Cox61} is known to be optimal 
among those non-anticipating scheduling policies that are non-preemptive, 
while the Gittins index rule \cite{Git89,Aal09QS,Aal11PEIS} is the best 
non-anticipating policy if preemption is allowed. For time-convex holding 
costs, van Mieghem~\cite{Mie95AAP} launched and studied the generalized 
$c\mu$ rule, which is a preemptive index policy serving always the customer 
with highest index $\gamma_k(x_i) \mu_k$, where $k$ refers to the class of 
customer~$i$, $x_i$ to the current total time that customer~$i$ has already 
spent in the system, and $\mu_k$ is the inverse of the mean service time for 
class-$k$ customers. He proved that the generalized $c\mu$ rule is 
asymptotically optimal in the heavy traffic limit for a single-server system. 
Mandelbaum and Stolyar \cite{Man04OpnsRes} generalized this result for 
multi-server systems.

While the generalized $c\mu$ rule is a near-optimal scheduling policy for 
patient customers under heavy traffic, it can still be improved in normal 
traffic conditions. This is demonstrated in \cite{Aal26QS}, where we 
applied the Whittle index approach \cite{Whi88JAP} to develop an index 
policy for patient customers. The developed Whittle index policy is always 
preemptively serving the customer with the highest index 
$W^\circ_k(x_i,y_i)$ defined by 
\begin{equation}
W^\circ_k(x_i,y_i) = 
\sup_{\Delta > 0} 
\frac{E[\gamma_k(x_i + S_i - y_i) \, 1_{\{S_i - y_i \le \Delta\}} \mid S_i > y_i]}
{E[\min\{S_i - y_i,\Delta\} \mid S_i > y_i]}, 
\label{eq:Whittle-index-patient-intro}
\end{equation}
where $k$ denotes the class of customer~$i$, $x_i$ the current total time 
that customer~$i$ has already spent in the system, $y_i$ the amount of 
service that customer~$i$ has attained until now, $S_i$ its total service 
time (still unknown for the scheduler), and $1_A$ the indicator function of 
event~$A$. Thus, the developed Whittle index policy utilizes information 
not only on the time that a customer has already spent in the system 
(as the generalized $c\mu$ rule) but also (and unlike the generalized 
$c\mu$ rule) on the attained service.

Ata and Tongarlak \cite{Ata13QS} considered the corresponding scheduling 
problem with time-convex holding costs and impatient customers, however, 
restricting themselves to the M/M/1 + M case. By studying the approximating 
Brownian control problem and solving the associated Bellman equations, they 
were able to produce index policies for their original scheduling problem. 

In this paper, we apply the Whittle index approach for the dynamic scheduling 
problem in a multi-class M/G/$N$ + M queue with time-convex holding costs and 
impatient customers that have generally distributed service times and 
exponentially distributed abandonment times. By utilizing the results in 
\cite{Aal26QS} derived for patient customers, we develop an index policy for impatient customers. The related Whittle index is given by 
\begin{equation}
\begin{split}
& 
W_k(x_i,y_i) = \sup_{\Delta > 0} 
\frac{1}{\theta_k} 
\frac{E\left[ 
\Big( \gamma_k(x_i + S_i - y_i) + d_k \theta_k \Big) 
1_{\{S_i - y_i < \min\{D^{\mathrm R}_i,\Delta\}\}}
\mid S_i > y_i \right]}
{E[\min\{D^{\mathrm R}_i,S_i - y_i,\Delta\} \mid S_i > y_i]}, 
\end{split}
\label{eq:Whittle-index-impatient-intro}
\end{equation}
where $k$ denotes the class of customer~$i$, $x_i$ the current total time 
that customer~$i$ has already spent in the system, $y_i$ its attained service 
until now, $S_i$ its total service time, $\theta_k$ the abandonment intensity 
for class~$k$, and $D^{\mathrm R}_i$ the remaining abandonment time of the 
customer, which is, due to the memoryless property of the exponential 
distribution, also exponentially distributed.

The rest of the paper is organized as follows. 
In Section~\ref{sec:discrete-discounted-problem}, we present the closed 
version\footnote{
Following \cite{Whi81AnnProb,Whi88JAP}, we call the dynamic setup with 
multi-class Poisson arrivals as the {\em open version} of the scheduling 
problem, while the {\em closed version} refers to the transient case where 
there is a finite number of heterogeneous customers in the system at time~$0$ 
and no new arrivals later on.}
of the discrete-time scheduling problem as a restless bandit problem. 
The relaxed version of the problem is introduced in 
Section~\ref{sec:relaxed-problem}, and the relaxed problem is shown to be 
indexable and the corresponding Whittle index is derived explicitly in 
Section~\ref{sec:Whittle-index-discrete}. Next, 
in Section~\ref{sec:whittle-index-continuous}, we move from the discrete-time 
setup back to the original scheduling problem in continuous time and give 
an explicit characterization of the Whittle index and define the Whittle 
index policy in this setting. 
The paper is concluded in Section~\ref{sec:conc}.

\section{Scheduling problem in discrete time with discounted costs}
\label{sec:discrete-discounted-problem}

We consider the following multi-server optimal scheduling problem in discrete 
time. Assume that there are $N$ homogeneous parallel servers. Let the time 
slots be indexed by $t \in \{1,2,\ldots\}$. At the beginning of the first time 
slot, there are $K$ customers, and no new customers arrive later on.

Let $S_k$ denote the random service time of customer~$k \in \{1,2,\ldots\,K\}$, 
which is assumed to be independent (of everything else). Let $\mu_k(n)$, 
$n \in \{1,2,\ldots\}$, denote the conditional probability that the service 
time of customer~$k$ is equal to $n$, given that it is strictly greater than 
$n - 1$, 
\begin{equation}
\mu_k(n) = P\{ S_k = n \mid S_k > n - 1 \}.
\label{eq:mu-definition}
\end{equation}
The service time distribution of customer~$k$ is otherwise general but, 
for mathematical simplicity, we assume that there is a finite value 
$N^\mu_k \in \{1,2,\ldots\}$ such that 
\begin{equation}
\mu_k(N^\mu_k) = 1.
\label{eq:math-simplicity-assumption-mu}
\end{equation}
Thus, $N^\mu_k$ denotes the maximum service time.

Let $D_k$ denote the random abandonment time of customer $k$, i.e., 
the time interval from the beginning until the customer leaves the 
system due to abandonment (unless its service has already been completed 
prior to that moment).  We assume here that $D_k$ is independently and 
geometrically distributed taking values in $\{1,2,\ldots\}$ with a fixed 
abandonment probability $\theta_k > 0$.

Customers are served according to a non-anticipating scheduling policy $\pi$ 
that allows preemptions. Let $\Pi$ denote the family of such disciplines. 
At the beginning of each time slot $t \in \{1,2,\ldots\}$, the scheduler 
chooses at most $N$ customers for service (during that time slot). Denote 
$A^\pi_k(t) = 1$ if customer~$k$ is chosen for service at the beginning of 
time slot $t$; otherwise let $A^\pi_k(t) = 0$. Thus, for any policy $\pi$ 
and time slot $t$, we have the following constraint: 
\begin{equation}
\sum_{k = 1}^K A^\pi_k(t) \le N.
\label{eq:capacity-constraint}
\end{equation}

Let $Y^\pi_k(t) \in \{0,1,\ldots,S_k-1\}$ denote the number of time slots 
that customer~$k$ has been served prior to time slot $t$ (i.e., its attained 
service) given that the customer is still in the system at the beginning of 
time slot $t$. Customer~$k$ leaves the system due to abandonment or service 
completion. If customer~$k$ is no longer in the system at the beginning of 
time slot $t$, we use symbol `$*$' and denote $Y^\pi_k(t) = *$. As for 
a possible abandonment, we assume that customer~$k$ leaves the system due to 
abandonment in time slot $t$ if the service of the customer has not been 
completed prior to time slot $t$ and $D_k = t$. Therefore, even if 
customer~$k$ was scheduled at the beginning of time slot $t$, it leaves 
the system due to abandonment if $D_k = t$.

Given that the customer is still in the system at the beginning of 
time slot~$t$, its holding costs accumulate at rate $\gamma_k(t)$ in that 
time slot. We assume that $\gamma_k(t)$ is a positive and increasing 
function of $t$, 
\begin{equation}
0 < \gamma_k(1) \le \gamma_k(2) \le \gamma_k(3) \le \cdots \; .
\label{eq:gamma-monotonicity}
\end{equation}
For mathematical simplicity, we further assume that there is a finite value 
$N^\gamma_k \in \{1,2,\ldots\}$ such that 
\begin{equation}
N^\gamma_k > N^\mu_k, \quad 
\gamma_k(N^\gamma_k) = \gamma_k(N^\gamma_k+1) = \gamma_k(N^\gamma_k+2) = 
\cdots \; .
\label{eq:math-simplicity-assumption-gamma}
\end{equation}

In addition to the holding costs, a penalty cost of size $d_k$ is incurred 
(as a lump sum) if customer~$k$ leaves the system due to abandonment. 
When the customer has left the system, no more costs accumulate any longer.

All costs are discounted in time by factor $\beta \in (0,1)$. The objective 
function in our scheduling problem is, thus, given by 
\begin{equation}
\left[ 
\sum_{t = 1}^\infty \sum_{k = 1}^K 
\beta^{t-1} \big( \gamma_k(t)  + d_k 1_{\{D_k = t \}} \big) 1_{\{Y^\pi_k(t) \ne * \}} 
\right].
\label{eq:expected-discrete-discounted-costs}
\end{equation}
The aim is to find the optimal scheduling policy $\pi \in \Pi$ that minimizes 
the expected discounted costs (\ref{eq:expected-discrete-discounted-costs}) 
subject to the capacity constraint (\ref{eq:capacity-constraint}) for all 
time slots $t$.

\section{Whittle index approach to the discrete-time problem}
\label{sec:relaxed-problem}

The discrete-time optimization problem described in the previous section 
belongs to the class of restless bandit problems, which are known to be 
mathematically intractable. In \cite{Whi88JAP}, Whittle proposed to relax 
such problems by first replacing the capacity constraint for all time slots 
by a time-averaged one, and then considering the Lagrangian version of the 
relaxed problem, which makes the problem separable. If the relaxed problem 
can be proved to be indexable, this leads to the so-called Whittle index rule, 
which solves the relaxed problem, and serves as a reasonable heuristic for the 
original restless bandit problem.

By applying the above-mentioned Whittle's approach, we get the following 
(separate) subproblems for each customer~$k$: 
Find the optimal policy~$\pi$ that minimizes the objective function 
\begin{equation}
f_{k,\beta}^{\pi} + \nu g_{k,\beta}^{\pi}, 
\label{eq:separable-discrete-discounted-costs}
\end{equation}
where $\nu$ can be interpreted as the unit price of work, $f_{k,\beta}^{\pi}$ 
as the expected discounted costs of customer~$k$, and $g_{k,\beta}^{\pi}$ 
as the expected discounted amount of work allocated to customer~$k$, 
\[
\begin{split}
& 
f_{k,\beta}^{\pi} = 
E\left[ 
\sum_{t = 1}^\infty 
\beta^{t-1} \big( \gamma_k(t)  + d_k 1_{\{D_k = t \}} \big) 1_{\{Y^\pi_k(t) \ne * \}} \right], \\
& 
g_{k,\beta}^{\pi} = 
E\left[ 
\sum_{t = 1}^\infty 
\beta^{t-1} A_k^\pi(t) \right].
\end{split}
\]
The separable subproblems (\ref{eq:separable-discrete-discounted-costs}) are 
considered in the context of Markov decision processes. The possible actions 
$a_k \in {\mathcal A} = \{0,1\}$ are ``to schedule'' ($a_k = 1$) and 
``not to schedule'' ($a_k = 0$).

Consider any customer~$k$. If the customer is still in the system, its state 
is described by the pair $(x,y)$, $x \ge y$, where $x$ refers to the total 
time that the customer has already spent in the system and $y$ to the amount 
of service that the customer has already attained when the decision is 
made.\footnote{
We assume that the decisions are made at the beginning of the time slot and 
customers leave the system at the end of the time slot in which their 
service requirement is completed. Thus, $x \in \{0,1,\ldots\}$ and 
$y \in \{0,1,\ldots,\min\{x,N^\mu_k-1\}\}$.}
If the customer is no longer in the system, the state is briefly denoted 
by symbol `$*$'. The set of all possible states is denoted by 
\[
{\mathcal S}_k = \{(x,y): x \in \{0,1,\ldots\}, y \in \{0,1,\ldots,\min\{x,N^\mu_k-1\}\}\} \cup \{*\}
\]
Let $q_k(\cdot|\cdot;a) \ge 0$ denote the corresponding state 
transition probability under action $a$. The non-zero state transition 
probabilities in our model are as follows: 
\begin{equation}
\begin{array}{ll}
q_k(x+1,y|x,y;0)   &= 1 - \theta_k, \\
q_k(x+1,y+1|x,y;1) &= (1 - \theta_k)(1-\mu_k(y+1)), \\
q_k(*|x,y;1)       &= \theta_k + (1 - \theta_k) \mu_k(y+1), \\
q_k(*|*,a)         &= 1, \quad a \in {\mathcal A}.
\end{array}
\label{eq:transition-probabilities-n}
\end{equation}
Note that the state $*$ is absorbing for any policy. Let then $c_k(z;a)$ 
denote the expected immediate costs in state~$z$ after action~$a$. 
In our model, for any $a \in {\mathcal A}$, 
\begin{equation}
\begin{array}{ll}
c_k(x,y;a) &= \gamma_k(x+1) + \theta_k d_k + a \nu, \\
c_k(*;a)   &= a \nu.
\end{array}
\label{eq:immediate-costs}
\end{equation}

Let $Z_k^\pi(t)$ denote the state of customer~$k$ at the beginning of time 
slot~$t$ under policy~$\pi$ and $V_{k,\beta}^{\pi}(z;\nu)$ the expected total 
discounted costs for this policy with a fixed initial state 
$z \in {\mathcal S}_k$: 
\begin{equation}
V_{k,\beta}^{\pi}(z;\nu) = 
E\left[ 
\sum_{t = 1}^\infty 
\beta^{t-1} \left( \big( \gamma_k(t)  + d_k 1_{\{D_k = t \}} \big) 1_{\{Z_k^\pi(t) \ne * \}} + \nu A_k^\pi(t) \right) 
\middle| Z_k^\pi(1) = z \right].
\label{eq:V-pi-function}
\end{equation}
In addition, define, for any state $z \in {\mathcal S}_k$, 
\begin{equation}
V_{k,\beta}(z;\nu) = \inf_\pi V_{k,\beta}^{\pi}(z;\nu), 
\label{eq:V-function}
\end{equation}
which is also called the value function of the problem. 
The policy $\pi_k^*$ for which 
\[
V_{k,\beta}^{\pi_k^*}(\cdot;\nu) = V_{k,\beta}(\cdot;\nu) 
\]
is said to be $(\nu,\beta)$-optimal for customer~$k$.

Since the state space~${\mathcal S}_k$ is discrete, the action 
space~${\mathcal A}$ is finite, and the immediate costs are bounded, 
the optimal policy belongs to the class of stationary policies 
\cite[Thm.~6.3]{Ros70}, which is denoted by $\Pi$. For each stationary 
policy $\pi \in \Pi$, the decisions $A_k^{\pi}(t)$ related to customer~$k$ 
at time~$t$ are deterministic depending just on the current state 
$Z_k^{\pi}(t)$ of the customer. In other words, each stationary policy 
$\pi$ is uniquely described by an activity set ${\mathcal B}_k^\pi$ 
that consists of the states in which customer~$k$ is scheduled to service. 
In addition, by \cite[Thm.~6.1]{Ros70}, we know that the value function 
$V_{k,\beta}(\cdot;\nu)$ satisfies the optimality equations: 
\begin{equation}
\begin{split}
& 
V_{k,\beta}(*;\nu) = 
\min\{ 0, \nu \} + \beta V_{k,\beta}(*;\nu), \\
& 
V_{k,\beta}(x,y;\nu) = 
\gamma_k(x+1) + \theta_k \big( d_k + \beta V_{k,\beta}(*;\nu) \big) \; + \\
& \quad 
\min \Big\{ 
\beta (1 - \theta_k) V_{k,\beta}(x+1,y;\nu), \\
& \quad \quad 
\nu + 
\beta (1 - \theta_k) \Big( 
\mu_k(y+1) V_{k,\beta}(*;\nu) + 
(1 - \mu_k(y+1)) V_{k,\beta}(x+1,y+1;\nu) 
\Big)
\Big\}.
\end{split}
\label{eq:opt-eqs-discrete-discounted-gen}
\end{equation}

Let us conclude this section by defining the indexability property, 
which is not automatically guaranteed for genuine restless bandit problems 
\cite{Whi88JAP}.

\begin{definition}
\label{def:indexability}
The relaxed optimization problem (\ref{eq:separable-discrete-discounted-costs}) 
related to customer~$k$ is indexable if, for any state 
$z \in {\mathcal S}_k$, 
there exists $W_{k,\beta}(z) \in [-\infty,\infty]$ such that 
\begin{itemize}
\item[(i)]
decision $a = 1$ (to schedule customer~$k$) is optimal in state~$z$ 
if and only if $\nu \le W_{k,\beta}(z)$; 
\item[(ii)]
decision $a = 0$ (not to schedule customer~$k$) is optimal in state~$z$ 
if and only if $\nu \ge W_{k,\beta}(z)$.
\end{itemize}
If the problem is indexable, the corresponding index $W_{k,\beta}(z)$ is 
called the Whittle index.
\end{definition}

Note that, according to this definition, the two actions are equally good 
(and, thus, optimal) in state $z$ if and only if $\nu = W_{k,\beta}(z)$.

\section{Whittle index for the discrete-time problem}
\label{sec:Whittle-index-discrete}

In this section, we solve the relaxed optimization problem with 
objective function (\ref{eq:separable-discrete-discounted-costs}). 
More precisely said, we prove that the relaxed problem is indexable 
and derive the corresponding Whittle index in 
Theorem~\ref{thm:Whittle-index-discrete-discounted}. 
Since we consider a single fixed customer~$k$, we leave out the related 
subscript~$k$ to lighten the notation.

\begin{theorem}
\label{thm:Whittle-index-discrete-discounted}
The relaxed optimization problem with objective function 
(\ref{eq:separable-discrete-discounted-costs}) is indexable, and the 
corresponding Whittle indexes are given by the following formulas: 
\begin{equation}
\begin{split}
& 
W_\beta(*) = 0, \\
& 
W_\beta(x,y) = 
\max \{ J_\beta(x,y;\Delta): \Delta \in \{1,\ldots,N^\mu-y\} \}, 
\quad 
(x,y) \in {\mathcal S} \setminus \{*\}, 
\end{split}
\label{eq:Whittle-index-discrete-discounted}
\end{equation}
where $J_\beta(x,y;\Delta)$ is defined by 
\begin{equation}
J_\beta(x,y;\Delta) = 
\frac{
\sum_{i=1}^{\Delta} \big( \beta (1 - \theta) \big)^{i-1} 
\bar \mu(y,i-1) \mu(y+i) \big( \gamma(x+i+1) + \theta d \big) }
{\sum_{i=1}^{\Delta} \big( \beta (1 - \theta) \big)^{i-1} \bar \mu(y,i-1)} \, 
\frac{\beta (1 - \theta)}{1 - \beta (1 - \theta)}, 
\label{eq:J-function-discounted}
\end{equation}
and the short-hand-notation $\bar \mu(y,i)$ by 
\[
\bar \mu(y,i) = \prod_{j=1}^i (1-\mu(y+j)).
\]
\end{theorem}

\paragraph{Proof} 
The main proof is given below in two parts ($1^\circ$ and $2^\circ$). 
In part $1^\circ$, we consider the non-negative values of $\nu$ 
($\nu \ge 0$). In this case, it is easy to see from the optimality 
equations~(\ref{eq:opt-eqs-discrete-discounted-gen}) that the minimum 
expected discounted cost $V_\beta(*;\nu)$ for the absorbing state $*$ 
clearly equals~$0$ and is achieved by the policies $\pi$ that choose 
action $0$ in state $*$. It follows that, for any $\nu \ge 0$, the 
optimality equations~(\ref{eq:opt-eqs-discrete-discounted-gen}) for 
the remaining states $(x,y) \in {\mathcal S} \setminus \{*\}$ read as 
follows: 
\begin{equation}
\begin{split}
& 
V_\beta(x,y;\nu) = 
\gamma(x+1) + \theta d + 
\min \Big\{ 
\beta (1 - \theta) V_\beta(x+1,y;\nu), \\
& \quad 
\nu + 
\beta (1 - \theta)(1 - \mu(y+1)) V_\beta(x+1,y+1;\nu) 
\Big\}.
\end{split}
\label{eq:opt-eqs-discrete-discounted-nu-pos}
\end{equation}

In the second part ($2^\circ$), we consider the non-positive values 
of $\nu$ ($\nu \le 0$). In this case, the minimum expected discounted 
cost $V_\beta(*;\nu)$ for the absorbing state $*$ clearly equals 
$\nu/(1 - \beta)$ and is achieved by those policies $\pi$ that choose 
action $1$ in state $*$. Moreover, for any $\nu \le 0$, the optimality 
equations~(\ref{eq:opt-eqs-discrete-discounted-gen}) for the remaining 
states $(x,y) \in {\mathcal S} \setminus \{*\}$ read as follows: 
\begin{equation}
\begin{split}
& 
V_\beta(x,y;\nu) = 
\gamma(x+1) + \theta \Big( d + \frac{\beta \, \nu}{1 - \beta} \Big) + 
\min \Big\{ 
\beta (1 - \theta) V_\beta(x+1,y;\nu), \\
& \quad 
\nu + 
\beta (1 - \theta) \Big( 
\mu(y+1) \frac{\nu}{1 - \beta} + 
(1 - \mu(y+1)) V_\beta(x+1,y+1;\nu) 
\Big)
\Big\}.
\end{split}
\label{eq:opt-eqs-discrete-discounted-nu-neg}
\end{equation}

Now we are ready to go through the main proof.

\vskip 6pt
$1^\circ$ 
Let $\nu \in [0,\infty)$. 
By defining $\tilde \gamma(x) = \gamma(x) + \theta d$ and 
$\tilde \beta = \beta (1 - \theta)$, the optimality equation 
(\ref{eq:opt-eqs-discrete-discounted-nu-pos}) is given by 
\begin{equation}
\begin{split}
& 
V_\beta(x,y;\nu) = 
\tilde \gamma(x+1) + 
\min \Big\{ 
\tilde \beta V_\beta(x+1,y;\nu), \\
& \quad 
\nu + 
\tilde \beta (1 - \mu(y+1)) V_\beta(x+1,y+1;\nu) 
\Big\}.
\end{split}
\label{eq:opt-eqs-discrete-discounted-nu-pos-tilde}
\end{equation}
Clearly, this takes exactly the same form as the optimality equation~(34) 
in \cite{Aal26QS}. Note further that $\tilde \beta \in (0,1)$ and function 
$\tilde \gamma(x)$ is positive and increasing.  Therefore, our claim for 
non-negative values $\nu \ge 0$ follows directly from the parts 
$1^\circ$--$3^\circ$ of the proof of Theorem~1 in \cite{Aal26QS}.

\vskip 6pt
$2^\circ$ 
Now, let $\nu \in (-\infty,0]$. 
We prove that the policy~$\pi$ with activity set 
\[
{\mathcal B}^\pi = {\mathcal S}, 
\]
according to which user~$k$ is scheduled in all states, is 
$(\nu,\beta)$-optimal. 
Due to the comments made just before equation 
(\ref{eq:opt-eqs-discrete-discounted-nu-neg}), 
it remains to prove that policy~$\pi$ is optimal in any state 
$(x,y) \in {\mathcal S} \setminus \{*\}$.

We start the proof by first deriving the value function $V_\beta^\pi(x,y;\nu)$ 
for policy~$\pi$ from the following equations: 
\begin{equation}
\begin{split}
& 
V_\beta^\pi(x,y;\nu) = 
\nu + \gamma(x+1) + \theta \Big( d + \frac{\beta \, \nu}{1 - \beta} \Big) \; + \\
& \quad 
\beta (1 - \theta) \Big( 
\mu(y+1) \frac{\nu}{1 - \beta} + (1 - \mu(y+1)) V_\beta^\pi(x+1,y+1;\nu) 
\Big), \\
& \quad \quad 
(x,y) \in {\mathcal S} \setminus \{*\}.
\end{split}
\label{eq:howard-eqs-discounted-case-2}
\end{equation}
By iterating these equations, we deduce that 
\begin{equation}
V_\beta^\pi(x,y;\nu) = 
\sum_{i=1}^{N^\mu-y} 
\big( \beta (1 - \theta) \big)^{i-1} 
\bar \mu(y,i-1) \big( \gamma(x+i) + \theta d \big) + 
\frac{\nu}{1 - \beta}.
\label{eq:howard-eqs-discounted-case-2-solution}
\end{equation}

Now, let $(x,y) \in {\mathcal S} \setminus \{*\}$. 
By (\ref{eq:howard-eqs-discounted-case-2-solution}), the following 
condition for optimality of $\pi$ in state~$(x,y)$ 
(based on (\ref{eq:opt-eqs-discrete-discounted-nu-neg})), 
\[
\begin{split}
& 
\beta (1 - \theta) V_\beta^\pi(x+1,y;\nu) \; \ge \\
& \quad 
\nu + 
\beta (1 - \theta) \Big( 
\mu(y+1) \frac{\nu}{1 - \beta} + 
(1 - \mu(y+1)) V_\beta^\pi(x+1,y+1;\nu) 
\Big), 
\end{split}
\]
can, by straightforward manipulations, be shown to be equivalent with condition 
\begin{equation}
\nu \le 
\beta (1 - \theta) 
\Big(
\sum_{i=1}^{N^\mu-y} 
\big( \beta (1 - \theta) \big)^{i-1} 
\bar \mu(y,i-1) \mu(y+i) \big( \gamma(x+i+1) + \theta d \big) 
\Big).
\label{eq:nu-req-discounted-case-2}
\end{equation}
However, this condition follows from our assumption that $\nu \le 0$ since 
the right-hand side of (\ref{eq:nu-req-discounted-case-2}) is clearly positive, 
which completes the proof of part~$2^\circ$.

\vskip 6pt
From $1^\circ$ and $2^\circ$ together, we deduce that the relaxed optimization 
problem with objective function (\ref{eq:separable-discrete-discounted-costs}) 
is indexable and the corresponding Whittle index is given by 
(\ref{eq:Whittle-index-discrete-discounted}).
\hfill $\Box$

\section{Whittle index for the continuous-time problem}
\label{sec:whittle-index-continuous}

In this section, we move from discounted to undiscounted costs. In addition, 
we move from the discrete-time setup to the original open version of the 
scheduling problem in continuous time, and determine the Whittle index policy 
in this setting.

Let us first consider undiscounted costs in the discrete-time model. 
Let $W_{k,1}(x,y)$ denote the Whittle index for customer~$k$ in state 
$(x,y) \in {\mathcal S} \setminus \{*\}$ related to the undiscounted costs, 
which is derived from the discounted-cost Whittle index $W_{k,\beta}(x,y)$ 
as follows: 
\[
W_{k,1}(x,y) = \lim_{\beta \to 1} W_{k,\beta}(x,y).
\]
By (\ref{eq:Whittle-index-discrete-discounted}), we have 
\[
W_{k,1}(x,y) = 
\max_{\Delta \in \{1,\ldots,N^\mu_k-y\}} 
\frac{
\sum_{i=1}^{\Delta} (1 - \theta_k)^i \bar \mu_k(y,i-1) \mu_k(y+i) 
\left( \frac{\gamma_k(x+i+1)}{\theta_k} + d_k \right)}
{\sum_{i=1}^{\Delta} (1 - \theta_k)^{i-1} \bar \mu_k(y,i-1)}, 
\]
which can be written as follows: 
\begin{equation}
\begin{split}
& 
W_{k,1}(x,y) = \max_{\Delta \in \{1,\ldots,N^\mu_k-y\}} 
\frac{E\left[ 
\left( \frac{\gamma_k(x + S_k - y + 1)}{\theta_k} + d_k \right) 
1_{\{D^{\mathrm R}_k > S_k - y\}} \, 1_{\{S_k - y \le \Delta\}} 
\mid S_k > y \right]}
{E[\min\{D^{\mathrm R}_k,S_k - y,\Delta\} \mid S_k > y]}, 
\end{split}
\label{eq:Whittle-index-discrete-undiscounted}
\end{equation}
where $D^{\mathrm R}_k$ refers to the remaining abandonment time of the 
customer, which is also geometrically distributed.

When considering the continuous-time model related to the original open 
version of the scheduling problem, the state of a customer in the system 
is described by the pair $(x,y)$ of real values $x \ge y \ge 0$, where 
$x$ denotes the current total time that the customer has already spent 
in the system and $y$ the amount of service that the customer has already 
attained until now. In addition, the holding cost rate function $\gamma_k(x)$ 
is defined for all non-negative real values $x \ge 0$ and assumed to be 
non-negative and increasing so that the accumulated holding costs are convex. 
Note also that, in this continuous-time model, the holding cost rate 
$\gamma_k(x)$ is given per time unit while the abandonment penalty $d_k$ 
is still a lump sum. Moreover, geometric abandonment times are replaced by 
exponentially distributed abandonment times so that $\theta_k$ now refers to 
the abandonment intensity per time unit (instead of the abandonment 
probability per time slot).

The continuous-time counterpart of the Whittle index 
(\ref{eq:Whittle-index-discrete-undiscounted}) is determined by letting the 
time slot shrink down to $0$ while, at the same time, scaling the service completion and abandonment probabilities accordingly together with letting 
$N^\mu_k \to \infty$ and $N^\gamma_k \to \infty$. As the result, the 
continuous-time Whittle index $W_k(x,y)$ for a {class-$k$} customer with 
generally distributed service times $S_k$, exponential abandonment times 
$D_k$, convexly increasing holding costs $\gamma_k(x)$, and abandonment 
penalties~$d_k$ reads as follows: 
\begin{equation}
\begin{split}
& 
W_k(x,y) = \sup_{\Delta > 0} 
\frac{E\left[ 
\left( \frac{\gamma_k(x + S_k - y)}{\theta_k} + d_k \right) 
1_{\{S_k - y < \min\{D^{\mathrm R}_k,\Delta\}\}}
\mid S_k > y \right]}
{E[\min\{D^{\mathrm R}_k,S_k - y,\Delta\} \mid S_k > y]}, 
\quad x \ge y \ge 0,
\end{split}
\label{eq:Whittle-index-continuous-undiscounted}
\end{equation}
where $D^{\mathrm R}_k$ refers to the remaining abandonment time of the 
customer, which is also exponentially distributed.

Next we define two functions for each customer class $k$, which are 
related to both the corresponding service time distribution and the 
cost rate function, 
\begin{equation}
\begin{split}
& H_k(x,y) = 
\frac{E\left[ 
\left( \frac{\gamma_k(x + S_k - y)}{\theta_k} + d_k \right) 
1_{\{S_k - y < D^{\mathrm R}_k\}}
\mid S_k > y \right]}
{E[\min\{D^{\mathrm R}_k,S_k - y\} \mid S_k > y]}, 
\quad x \ge y \ge 0, \\
& 
h_k(x,y) = \left( \frac{\gamma_k(x)}{\theta_k} + d_k \right) \mu_k(y), 
\quad x \ge y \ge 0, 
\end{split}
\label{eq:Hh-functions-continuous-undiscounted}
\end{equation}
where $\mu_k(y)$ is the continuous-time hazard rate function of the service 
time $S_k$. These two functions can be utilized when giving a lower bound 
for the Whittle index (\ref{eq:Whittle-index-continuous-undiscounted}).

\begin{proposition}
\label{prop:W-lower-bound}
For all $x \ge y \ge 0$, 
\begin{equation}
W_k(x,y) \ge \max\{H_k(x,y),h_k(x,y)\}.
\label{eq:W-lower-bound}
\end{equation}
\end{proposition}

\paragraph{Proof} 
The proof is presented in Appendix.
\hfill $\Box$
\vskip 12pt

Now we give two additional results, which are related to the service time 
distribution classes IHR (i.e., the hazard rate function $\mu_k(y)$ is 
increasíng for all $y$) and DHR (i.e., the hazard rate function $\mu_k(y)$ is 
decreasíng for all $y$), respectively.

\begin{proposition}
\label{prop:IHR}
If $H_k(x+\Delta,y+\Delta)$ is increasing (w.r.t.\ $\Delta$) for all 
$x \ge y \ge 0$ and $\Delta \ge 0$, then 
\begin{equation}
W_k(x,y) = H_k(x,y) 
\quad \hbox{for all $x \ge y \ge 0$}.
\label{eq:IHR}
\end{equation}
In particular, a sufficient condition for this is to require that 
service time $S_k$ has an IHR distribution.
\end{proposition}

\paragraph{Proof} 
The proof is presented in Appendix.
\hfill $\Box$
\vskip 12pt

\begin{proposition}
\label{prop:DHR}
If $h_k(x+\Delta,y+\Delta)$ is decreasing (w.r.t.\ $\Delta$) for all 
$x \ge y \ge 0$ and $\Delta \ge 0$, then 
\begin{equation}
W_k(x,y) = h_k(x,y).
\quad \hbox{for all $x \ge y \ge 0$}, 
\label{eq:DHR}
\end{equation}
In particular, a necessary condition for this is to require that 
service time $S_k$ has a DHR distribution.\footnote{
Note that Proposition~\ref{prop:DHR} is consistent with Equation~(32) in 
\cite{Aal24QS}, where we assumed linear holding costs with $\gamma_k(x) = c_k$ 
for all $x$.}

\end{proposition}

\paragraph{Proof} 
The proof is presented in Appendix.
\hfill $\Box$
\vskip 12pt

Let us now finally define the Whittle index policy in this continuous-time setting.

\begin{definition}
\label{def:Whittle-index-policy}
Consider the original continuous-time scheduling problem with $N$ homogeneous 
parallel servers, convex holding costs, and exponentially distributed 
abandonment times. At any time~$t$, the Whittle index policy chooses to serve 
\begin{itemize}
\item[(i)]
all the customers, 
if there are at most $N$ customers in the system; 
\item[(ii)]
those $N$ customers that have the highest Whittle indexes $W_k(x,y)$ 
given in (\ref{eq:Whittle-index-continuous-undiscounted}), 
if there are more than $N$ customers in the system.
\end{itemize}
\end{definition}

\section{Some concluding remarks}
\label{sec:conc}

In this paper, we applied the Whittle index approach to the dynamic scheduling 
problem in a multi-class M/G/$N$ + M queue with convex delay costs and 
impatient customers that have exponential abandonment times. This last 
assumption allowed us make a very straightforward use of the results of the 
corresponding problem with patient customers presented in \cite{Aal26QS}.

As for the case with customers that have non-exponential abandonment times, 
it is no longer possible to directly utilize the results in \cite{Aal26QS}. 
In a recent preprint \cite{Aal26QS-preprint}, we study the dynamic scheduling 
problem in a multi-class M/G/$N$ + G queue and demonstrate how to apply the 
Whittle index approach when delay costs are linear (i.e., less general than 
convex) and abandonment times have IHR distributions (i.e., more general 
than exponential).

\newpage



\section*{Appendix}

This appendix gives the proofs of 
Propositions~\ref{prop:W-lower-bound}-\ref{prop:DHR}. 
Therefore, we consider a single class-$k$ customer in the same continuous-time 
setup as in the latter part of Section~\ref{sec:whittle-index-continuous}. 
The related subscript~$k$ is again left out to lighten the notation.

Assume that the cost rate function $\gamma(x)$ is a non-negative and increasing 
function of $x$ so that the accumulated holding costs are convex. In addition, 
let $f(y)$, $\bar F(y)$, and $\mu(y)$ denote, respectively, the density 
function, the tail distribution function, and the hazard rate function of 
the generally distributed service time $S$, 
\[
\bar F(y) = P\{ S > y \} = 1 - \int_0^y f(t) \, \mathrm{d}t, \quad 
\mu(y) = \frac{f(y)}{\bar F(y)}, \quad y \ge 0.
\]
In addition, let $\bar G(x)$ denote the tail distribution function of the 
exponentially distributed abandonment time $D$ (as well as the remaining 
abandonment time $D^{\mathrm R}$), 
\[
\bar G(x) = P\{ D > x \} = P\{ D^{\mathrm R} > x \} = e^{-\theta x}, \quad x \ge 0.
\]
Note also that 
\[
\bar G(x+\Delta) = \bar G(x) \bar G(\Delta), \quad x, \Delta \ge 0.
\]

The efficiency function $J(x,y;\Delta)$ related to the service time 
distribution is defined, for any $x \ge y \ge 0$ and $\Delta > 0$, by 
\begin{equation}
\begin{split}
J(x,y;\Delta) 
& = 
\frac{E\left[ 
\left( \frac{\gamma(x + S - y)}{\theta} + d \right) 
1_{\{S - y < \min\{D^{\mathrm R},\Delta\}\}}
\mid S > y \right]}
{E[\min\{D^{\mathrm R},S - y,\Delta\} \mid S > y]} \\
& = 
\frac{\int_0^\Delta \left( \frac{\gamma(x+t)}{\theta} + d \right) f(y+t) \, \bar G(t) \, \mathrm{d}t}
{\int_0^\Delta \bar F(y+t) \, \bar G(t) \, \mathrm{d}t}.
\end{split}
\label{eq:app-J-function}
\end{equation}
It follows that the functions $W(x,y)$, $H(x,y)$, and $h(x,y)$ 
defined in Section~\ref{sec:whittle-index-continuous} 
(see Equations~(\ref{eq:Whittle-index-continuous-undiscounted}) and 
(\ref{eq:Hh-functions-continuous-undiscounted})) can be expressed as follows: 
\begin{equation}
W(x,y) = \sup_{\Delta > 0} J(x,y;\Delta), \quad 
H(x,y) = \lim_{\Delta \to \infty} J(x,y;\Delta), \quad 
h(x,y) = \lim_{\Delta \to 0} J(x,y;\Delta).
\label{eq:app-WHh-functions}
\end{equation}
In addition, by (\ref{eq:Hh-functions-continuous-undiscounted}), we have 
\begin{equation}
J(x,y;\Delta) = 
\frac{\int_0^\Delta h(x+t,y+t) \, \bar F(y+t) \, \bar G(t) \, \mathrm{d}t}
{\int_0^\Delta \bar F(y+t) \, \bar G(t) \, \mathrm{d}t}.
\label{eq:app-J-function-alt}
\end{equation}

Now we ready to prove Propositions~\ref{prop:W-lower-bound}-\ref{prop:DHR}.

\subsection*{\em Proof of Proposition~\ref{prop:W-lower-bound}}
\label{subsec:app-proof-lower-bound}
Proposition~\ref{prop:W-lower-bound} claims that 
\begin{equation}
W(x,y) \ge \max\{H(x,y),h(x,y)\}.
\label{eq:app-proof-lower-bound}
\end{equation}
This claim follows immediately from (\ref{eq:app-WHh-functions}).
\hfill $\Box$

\subsection*{\em Proof of Proposition~\ref{prop:IHR}}
\label{subsec:app-proof-IHR}
Proposition~\ref{prop:IHR} claims that, if $H(x+\Delta,y+\Delta)$ is increasing 
w.r.t.\ $\Delta$, then 
\begin{equation}
W(x,y) = H(x,y).
\label{eq:app-IHR}
\end{equation}
This claim follows clearly from (\ref{eq:app-WHh-functions}) and the following 
result, which is proved below: 
\begin{equation}
H(x+\Delta,y+\Delta) \ge H(x,y) \Longleftrightarrow 
H(x,y) \ge J(x,y;\Delta).
\label{eq:app-IHR-1}
\end{equation}
By (\ref{eq:app-WHh-functions}) and (\ref{eq:app-J-function-alt}), we have 
\[
\begin{split}
&
H(x+\Delta,y+\Delta) \ge H(x,y) \\
& \Longleftrightarrow \quad 
\frac{\int_0^\infty h(x+\Delta+t,y+\Delta+t) \, \bar F(y+\Delta+t) \, \bar G(\Delta+t) \, \mathrm{d}t}
{\int_0^\infty \bar F(y+\Delta+t) \, \bar G(\Delta+t) \, \mathrm{d}t} \ge 
\frac{\int_0^\infty h(x+t,y+t) \, \bar F(y+t) \, \bar G(t) \, \mathrm{d}t} 
{\int_0^\infty \bar F(y+t) \, \bar G(t) \, \mathrm{d}t} \\
& \Longleftrightarrow \quad 
\left( 
\int_0^\infty h(x+\Delta+t,y+\Delta+t) \, \bar F(y+\Delta+t) \, \bar G(\Delta+t) \, \mathrm{d}t
\right) 
\left( 
\int_0^\infty \bar F(y+t) \, \bar G(t) \, \mathrm{d}t
\right) \; \ge \\
& \quad \quad \quad \quad \quad 
\left( 
\int_0^\infty h(x+t,y+t) \, \bar F(y+t) \, \bar G(t) \, \mathrm{d}t
\right) 
\left( 
\int_0^\infty \bar F(y+\Delta+t) \, \bar G(\Delta+t) \, \mathrm{d}t
\right) \\
& \Longleftrightarrow \quad 
\left( 
\int_0^\infty h(x+\Delta+t,y+\Delta+t) \, \bar F(y+\Delta+t) \, \bar G(\Delta+t) \, \mathrm{d}t
\right) 
\left( 
\int_0^\infty \bar F(y+t) \, \bar G(t) \, \mathrm{d}t
\right) \; \ge \\
& \quad \quad \quad \quad \quad 
\left( 
\int_0^\infty h(x+t,y+t) \, \bar F(y+t) \, \bar G(t) \, \mathrm{d}t
\right) 
\left( 
\int_0^\infty \bar F(y+t) \, \bar G(t) \, \mathrm{d}t - 
\int_0^\Delta \bar F(y+t) \, \bar G(t) \, \mathrm{d}t
\right) \\
& \Longleftrightarrow \quad 
\left( 
\int_0^\infty h(x+t,y+t) \, \bar F(y+t) \, \bar G(t) \, \mathrm{d}t
\right) 
\left( 
\int_0^\Delta \bar F(y+t) \, \bar G(t) \, \mathrm{d}t
\right) \; \ge \\
& \quad \quad \quad \quad \quad 
\left( 
\int_0^\infty h(x+t,y+t) \, \bar F(y+t) \, \bar G(t) \, \mathrm{d}t \; - 
\right. \\
& \quad \quad \quad \quad \quad \quad \quad 
\left. 
\int_0^\infty h(x+\Delta+t,y+\Delta+t) \, \bar F(y+\Delta+t) \, \bar G(\Delta+t) \, \mathrm{d}t 
\right) 
\left( 
\int_0^\infty \bar F(y+t) \, \bar G(t) \, \mathrm{d}t 
\right) \\
& \Longleftrightarrow \quad 
\left( 
\int_0^\infty h(x+t,y+t) \, \bar F(y+t) \, \bar G(t) \, \mathrm{d}t
\right) 
\left( 
\int_0^\Delta \bar F(y+t) \, \bar G(t) \, \mathrm{d}t
\right) \; \ge \\
& \quad \quad \quad \quad \quad 
\left( 
\int_0^\Delta h(x+t,y+t) \, \bar F(y+t) \, \bar G(t) \, \mathrm{d}t
\right) 
\left( 
\int_0^\infty \bar F(y+t) \, \bar G(t) \, \mathrm{d}t 
\right) \\
& \Longleftrightarrow \quad 
\frac{\int_0^\infty h(x+t,y+t) \, \bar F(y+t) \, \bar G(t) \, \mathrm{d}t} 
{\int_0^\infty \bar F(y+t) \, \bar G(t) \, \mathrm{d}t} 
\ge 
\frac{\int_0^\Delta h(x+t,y+t) \, \bar F(y+t) \, \bar G(t) \, \mathrm{d}t} 
{\int_0^\Delta \bar F(y+t) \, \bar G(t) \, \mathrm{d}t} \\
& \Longleftrightarrow \quad 
H(x,y) \ge J(x,y;\Delta),
\end{split}
\]
which justifies (\ref{eq:app-IHR-1}).

By applying (\ref{eq:app-WHh-functions}) and (\ref{eq:app-J-function-alt}) and taking the derivative of $H(x+\Delta,y+\Delta)$ w.r.t.\ $\Delta$, where 
\[
H(x+\Delta,y+\Delta) = 
\frac{\int_0^\infty h(x+\Delta+t,y+t) \, \bar F(y+\Delta+t) \, \bar G(t) \, \mathrm{d}t}
{\int_0^\infty \bar F(y+\Delta+t) \, \bar G(t) \, \mathrm{d}t}, 
\]
it is easy to show that $H(x+\Delta,y+\Delta)$ is an increasing function of 
$\Delta$ whenever $h(x+\Delta,y+\Delta)$ is such. Thus, the other claim of 
Proposition~\ref{prop:IHR} follows from our assumption that the holding cost 
rate function $\gamma(x)$ is non-negative and increasing for all~$x$. 
\hfill $\Box$

\subsection*{\em Proof of Proposition~\ref{prop:DHR}}
\label{subsec:app-proof-DHR}
Proposition~\ref{prop:DHR} claims that, if $h(x+\Delta,y+\Delta)$ is 
decreasing w.r.t.\ $\Delta$, then 
\begin{equation}
W(x,y) = h(x,y). 
\label{eq:app-proof-DHR}
\end{equation}
Note first that, by (\ref{eq:app-J-function-alt}), we have 
\begin{equation}
\frac{\partial}{\partial \Delta} J(x,y;\Delta) = 
\frac{\bar F(y+\Delta) \, \bar G(\Delta) 
\big( h(x+\Delta,y+\Delta) - J(x,y;\Delta) \big)}
{\int_0^\Delta \bar F(y+t) \, \bar G(t) \, \mathrm{d}t}.
\label{eq:app-proof-DHR-1}
\end{equation}
On the other hand, if $h(x+\Delta,y+\Delta)$ is decreasing w.r.t.\ $\Delta$, 
it follows from (\ref{eq:app-J-function-alt}) that 
\begin{equation}
J(x,y;\Delta) = 
\frac{\int_0^\Delta h(x+t,y+t) \, \bar F(y+t) \, \bar G(t) \, \mathrm{d}t}
{\int_0^\Delta \bar F(y+t) \, \bar G(t) \, \mathrm{d}t} \ge 
h(x+\Delta,y+\Delta).
\label{eq:app-proof-DHR-2}
\end{equation}
Thus, by (\ref{eq:app-proof-DHR-1}), $J(x,y;\Delta)$ is decreasing 
w.r.t.\ $\Delta$ under this condition, which justifies our claim 
(\ref{eq:app-proof-DHR}) due to (\ref{eq:app-WHh-functions}).

Since 
\[
h(x+\Delta,y+\Delta) = \left( \frac{\gamma(x+\Delta)}{\theta} + d \right) \mu(y+\Delta) 
\] 
by (\ref{eq:Hh-functions-continuous-undiscounted}), the other claim of 
Proposition~\ref{prop:DHR} follows from our assumption that the holding 
cost rate function $\gamma(x)$ is non-negative and increasing for all~$x$.
\hfill $\Box$

\newpage


\bibliographystyle{plain} 
\bibliography{whittle-abandonment-convex} 

\end{document}